\documentclass[journal]{IEEEtran}

\usepackage{ifpdf}
\usepackage{xcolor}
\usepackage{textcomp}
\usepackage{leading}
\usepackage{calligra}
\definecolor{salmon}{RGB}{250,128,114}
\definecolor{lightsteelblue}{RGB}{176,196,222}
\usepackage{dutchcal}

\usepackage{cite}

\ifCLASSINFOpdf
   \usepackage[pdftex]{graphicx}
\else
   \usepackage[dvips]{graphicx}
\fi

\usepackage{amsmath}
\usepackage{amsthm}
\usepackage{amsfonts}
\usepackage{amssymb}
\usepackage{graphicx}
\usepackage{calrsfs}
\usepackage{mathrsfs}
\usepackage{multirow}
\usepackage[mathscr]{eucal}
\usepackage{mathtools}

\usepackage{algorithm}
\usepackage{algorithmic}

\usepackage{array}
\usepackage{mdwtab}
\usepackage{tabularray}

\ifCLASSOPTIONcompsoc  \usepackage[caption=false,font=normalsize,labelfont=sf,textfont=sf]{subfig}
\else
  \usepackage[caption=false,font=footnotesize]{subfig}
\fi

\usepackage{float}
\usepackage{stfloats}

\usepackage{url}

\usepackage{tikz}
\usetikzlibrary{arrows,shapes,backgrounds,circuits.ee.IEC,circuits.logic.US}
\usetikzlibrary {calc,positioning,shapes.misc}
\usetikzlibrary{%
	arrows,%
	shapes.misc,
	shapes.arrows,%
	chains,%
	matrix,%
	positioning,
	scopes,%
	decorations.pathmorphing,
	shadows%
}

\newtheorem{Ex}{Example}[]
\newtheorem{Def}{Definition}[]
\newtheorem{Prop}{Proposition}[]
\newtheorem{Teo}{Theorem}[]

\newtheorem{Obs}{Remark}[]
\newtheorem{Lema}{Lemma}[]

\newcommand{\Z}{\mathbb{Z}}

\newcommand{\R}{\mathbb{R}}

\newcommand{\vor}{\mathscr{V}}

\newcommand{\cod}{\mathcal{C}}

\newcommand{\ii}{i}

\begin{document}

\title{Achieving uniform side information gain with multilevel lattice codes over the ring of integers}

\author{Juliana~G. F. Souza ,~\IEEEmembership{Student Member,~IEEE}
        and~Sueli~I. R. Costa,~\IEEEmembership{Member,~IEEE,}
       \thanks{Manuscript received Month XX, 2025.}%
\thanks{Juliana G. F. Souza and Sueli I. R. Costa are with the Institute of Mathematics, Statistics and Scientific Computing (IMECC), University of Campinas (Unicamp), Campinas, São Paulo 13083-859, Brazil (e-mail: julianagfs@ime.unicamp.br / sueli@unicamp.br).}}


\maketitle

\begin{abstract}

The index coding problem aims to optimise broadcast communication by taking advantage of receiver-side information to improve transmission efficiency. In this letter, we explore the application of Construction $\pi_A$ lattices to index coding. We introduce a coding scheme, named \textit{CRT lattice index coding}, using Construction $\pi_A$ over $\Z$ to address the index coding problem. It is derived an upper bound for side information gain of a CRT lattice index code and conditions for the uniformity of this gain. The efficiency of this approach is shown through theoretical analysis and code design examples.

\end{abstract}


\section{Introduction}

\IEEEPARstart{T}{he} index coding problem, introduced by \textit{Birk and Kol} \cite{birk1998informed}, exploits receiver side information to optimize broadcast efficiency. It aims to encode messages, meeting all receiver's demands at the highest possible rate. Researchers have studied the capacity region of the AWGN broadcast channel \cite{tuncel2006slepian, natarajan2018lattice}, and the design of codes that exploit side information for transmission gains \cite{manesh2016, natarajan2015lattice, huang2017lattice, huang2018layered} in the scenario where the transmitter does not have knowledge of each receiver side information. 

Lattice-based strategies have been effective in communication scenarios, including achieving AWGN channel capacity and good performance on wiretap transmissions, and relay networks, among others. The use of lattices for index coding was first proposed by \textit{Natarajan et al.} \cite{natarajan2015lattice, natarajan2015lattice2}, where nested lattice codes over principal ideal domains (PIDs) were considered, assuming that all receivers demand all messages. Extensions of this approach include lattices from algebraic number fields \cite{huang2017lattice} and from cyclic division algebras \cite{huang2018layered}.

In this letter, we propose an index coding scheme using Construction $\pi_A$ lattices over $\mathbb{Z}$, called \textit{CRT lattice index coding}, focusing on code designs to provide uniform side information gain. It is somehow natural to explore this construction in this context, as mentioned in \cite{huang2017construction, huang2018lattices, jsouza2024multilevel}. Construction $\pi_A$ lattice is a special case of Construction $A$ from codes over rings that can be factorised as a product of rings by the Chinese Remainder Theorem (CRT), providing lattice codes with the benefit of multilevel encoding and multistage decoding with reduced decoding complexity.

The remainder of this work is organised as follows. Section II covers preliminaries on lattices and lattice codes. We introduce the index coding problem in Section III. In Section IV it is proposed the CRT lattice index coding construction and derived an upper bound for its side information gain when this construction is obtained from same rank codes. Section V presents cases in which uniform side information gain is achieved. Conclusions and perspectives are described in Section VI.

\section{Preliminaries}
\subsection{Lattices and lattice codes}
We summarise next some concepts and properties related to lattice codes \cite{Con2013, Sue2018}. In this work, the distance and norm considered are Euclidean.

A \textit{lattice} $\Lambda$ is a discrete subset of $\mathbb{R}^n$ generated by all linear integer combinations of a set of independent vectors $\{v_1,\ldots,v_m\}$. $m$ is called the lattice \textit{rank} and we deal here only with full-rank lattices ($m=n$). A matrix $B$ whose columns are the vectors $v_i$ is called a \textit{generator matrix} of $\Lambda$. The \textit{volume} of a (full-rank) lattice is given by $\text{vol}(\Lambda)=|\det(B)|$.

Given a point $z\in\mathbb{R}^n$ and a lattice $\Lambda\subset\mathbb{R}^n$, we define $Q_{\Lambda}(z)$ as the \textit{closest lattice point} to $z$, also called the quantization of $z$,
\begin{equation}
    Q_{\Lambda}(z) = x\in\Lambda; \ ||z-x||\leq||z-y|| \quad \forall y\in\Lambda,
\end{equation}
\noindent and ties must be chosen. \textit{The Voronoi region} of a lattice $\Lambda$, $\vor_\Lambda$, is the set of all points in $\mathbb{R}^n$ that are mapped to the origin under $Q_{\Lambda}$, $\vor_{\Lambda}(0)=\vor_\Lambda = \{ z\in\R^n; \|z\|\leq\|z-y\| \ \forall y\in\Lambda\}.$

Considering a sublattice $\Lambda'\subset \Lambda$, the set $\Lambda/\Lambda'$ is called a \textit{Voronoi constellation} and it is associated with the elements of $\Lambda$ inside the Voronoi region of $\Lambda'$.

The \textit{minimum distance} $d_{\text{min}}(\Lambda)$ and the \textit{centre density} $\delta(\Lambda)$ of a lattice $\Lambda$ are defined, respectively, as
\begin{equation}
    d_{\text{min}}(\Lambda)=\min_{0\neq x\in\Lambda}||x|| \ \ \text{and} \ \ \delta(\Lambda) = \dfrac{\left(d_{\min}(\Lambda)/2\right)^n}{\text{vol}(\Lambda)}.
\end{equation}

The \textit{kissing number}, $K(\Lambda)$, of a lattice is the number of lattice points having the minimum non-vanishing norm.

A \textit{linear code} over $\mathbb{Z}_q$, the ring of integers modulo $q$, is a subset $\cod\subset\mathbb{Z}_q^n$ which is an additive subgroup of $\mathbb{Z}_q^n$. The rank of a code, $\text{rank}(\cod)$, is the minimum number of generators of $\cod$. For $q=p$ (prime number) a linear code is a vector subspace of dimension $k\leq n$, called a $(n,k)-$linear code. A code $\cod\subset\mathbb{Z}_q^n$ is said to be a \textit{free linear code} if it has a basis over $\Z_q$.

A method for obtaining lattices from linear codes is the \textit{Construction A}, \cite{Con2013, Sue2018}. Consider the mapping $\rho: \mathbb{Z}\rightarrow\mathbb{Z}_q$, the natural modulo $q$ reduction and $\sigma: \mathbb{Z}_q\rightarrow\mathbb{Z}$, the standard inclusion map, extended to vectors component-wise. Given a linear code $\cod\subset\mathbb{Z}_q^n$, the Construction A lattice associated with $\cod$, denoted by $\Lambda_A(\cod)$, is defined as
\begin{equation}
    \Lambda_A(\cod) = \rho^{-1}(\cod) = \sigma(\cod) + q\mathbb{Z}^n.
    \label{constructionA}
\end{equation}

It is shown that $\Lambda_A(\cod)$ is a full-rank lattice, $q\mathbb{Z}^n\subset\Lambda_A(\cod)\subset\mathbb{Z}^n$, $\text{vol}(\Lambda)=q^n/|\cod|$, where $|\cod|$ is the cardinality of the code $\cod$ and $d_{\min}(\Lambda_A(\cod))=\min\{d_{\min}(\cod),q\}$, \cite{Sue2018}.

\textit{Huang and Narayanan} introduced Construction $\pi_A$ lattices, in \cite{huang2017construction}. It is a special case of Construction A and generates lattices over various integer rings.
It relies on a ring isomorphism given by the Chinese Remainder Theorem (CRT) \cite[(23)]{huang2017construction}, \cite{huang2017lattice, huang2018layered, jsouza2024multilevel}.

\begin{Def}[Construction $\pi_A$ lattice, \cite{huang2017construction}]
	\label{construcao_pia}
	Let $p_1,\ldots,p_r$ be distinct primes and let $q=\prod_{j=1}^{r}p_j$. Consider $l_j, n$ to be integers such that $l_j\leq n$ and let $G_j$ be a generator matrix of an $(n,l_j)$-linear code over $\mathbb{Z}_{p_j}$ for $j\in\{1,\ldots,r\}$. Construction $\pi_A$ consists of the following steps: 1) Define the discrete codebooks $\mathcal{C}_j=\{\mathbf{x}=G_j \cdot u: u \in \mathbb{Z}_{p_j}^{l_j}\}$ for $j \in \{1,\ldots,r\}$. 2) Construct $\mathcal{C}=\phi^{-1}(\mathcal{C}_1\times\ldots\times\mathcal{C}_r)$ where $\phi^{-1}:\mathbb{Z}_{p_1}^n\times\cdots\times\mathbb{Z}_{p_r}^n \rightarrow \mathbb{Z}_q^n$ is the CRT ring isomorphism. 
    3) Tile $\mathcal{C}$ to the entire $\mathbb{R}^n$ to form $\Lambda_{\pi_A}(\mathcal{C})=\sigma(\mathcal{C})+q\mathbb{Z}^n$.
\end{Def}

Due to their multilevel structure, Construction $\pi_A$ lattices allow multistage decoding, where each class representative is decoded sequentially at each level.

\section{Problem Statement}

Consider a scenario with $r$ independent messages $w_1, \ldots, w_r$ from alphabets $W_1, \ldots, W_r$. These messages are jointly encoded into a codeword $x = f(w_1, \ldots, w_r) \in \cod$, where $\cod \subset \mathbb{R}^n$ is an $n$-dimensional constellation. The received signal at receiver $l$, $l=1,...,L$ is
\[
y_l = x + z_l, \quad z_l \sim N(0, \sigma^2).
\]

Each receiver $l$ knows a subset $w_{S_l} = \{w_j : j \in S_l\}$ of the messages, with $S_l \subset \{1, \ldots, r\}$, and estimates $(\hat{w}_1^{(l)}, \ldots, \hat{w}_r^{(l)})$ of $(w_1, \ldots, w_r)$ using $y_l$ and $w_{S_l}$. This setup is illustrated in Figure~\ref{indexscheme}.

The rates of the $j$-th message, the side information rate, and the transmission rate are,
\(
R_j = \frac{1}{n} \log_2 |W_j| \ \text{bits/dim}, \quad R_{S_l} = \sum_{j \in S_l} R_j \ \text{bits/dim}, \quad R = \sum_{j=1}^r R_j \ \text{bits/dim},\) respectively.

Let $\cod_{S_l}$ denote the subcode of $\cod$ restricted to $w_{S_l}$. The minimum Euclidean distance in $\cod$ is denoted by $d_0$, and in $\cod_{S_l}$ by $d_{S_l}$. 
The squared minimum distance gain for $w_{S_l}$ is defined as,
\(
10 \log_{10} \left(d_{S_l}^2 / d_0^2\right) \ \text{dB}.
\) The side information gain at receiver $l$ is given by \cite{natarajan2015lattice, huang2018lattices},
\[
\Gamma(\cod, S_l) = \frac{10 \log_{10} \left(d_{S_l}^2 / d_0^2\right)}{R_{S_l}} \ \text{dB/bit/dim},
\]
and the overall gain is, 
\(
\Gamma(\cod) = \min_{S_l} \Gamma(\cod, S_l) \ \text{dB/bit/dim}.
\)
We say that a lattice index code provides \textit{uniform gain} if $\Gamma(\cod, S_l)$ is constant for all $S_l$, \cite{natarajan2015lattice}. 

At high SNR, the error probability for a lattice $\Lambda$ with minimum Euclidean distance $d_{\min}(\Lambda)$ and kissing number $K(\Lambda)$ is approximated by,
{\small \(
P_e(\Lambda, \sigma^2) \approx K(\Lambda) \cdot Q\left(\sqrt{\frac{d_{\min}^2(\Lambda)}{4 \sigma^2}}\right),
\)}
 where $Q(\cdot)$ is the Gaussian $Q$-function. The symbol error rate (SER) is, \(
\text{SER}(\Lambda, \sigma^2) \approx \frac{1}{n} P_e(\Lambda, \sigma^2)
\), \cite{zamir2014lattice}.

\begin{figure}[t]
\centering
\vspace{-0.5cm}
\includegraphics[scale=.25]{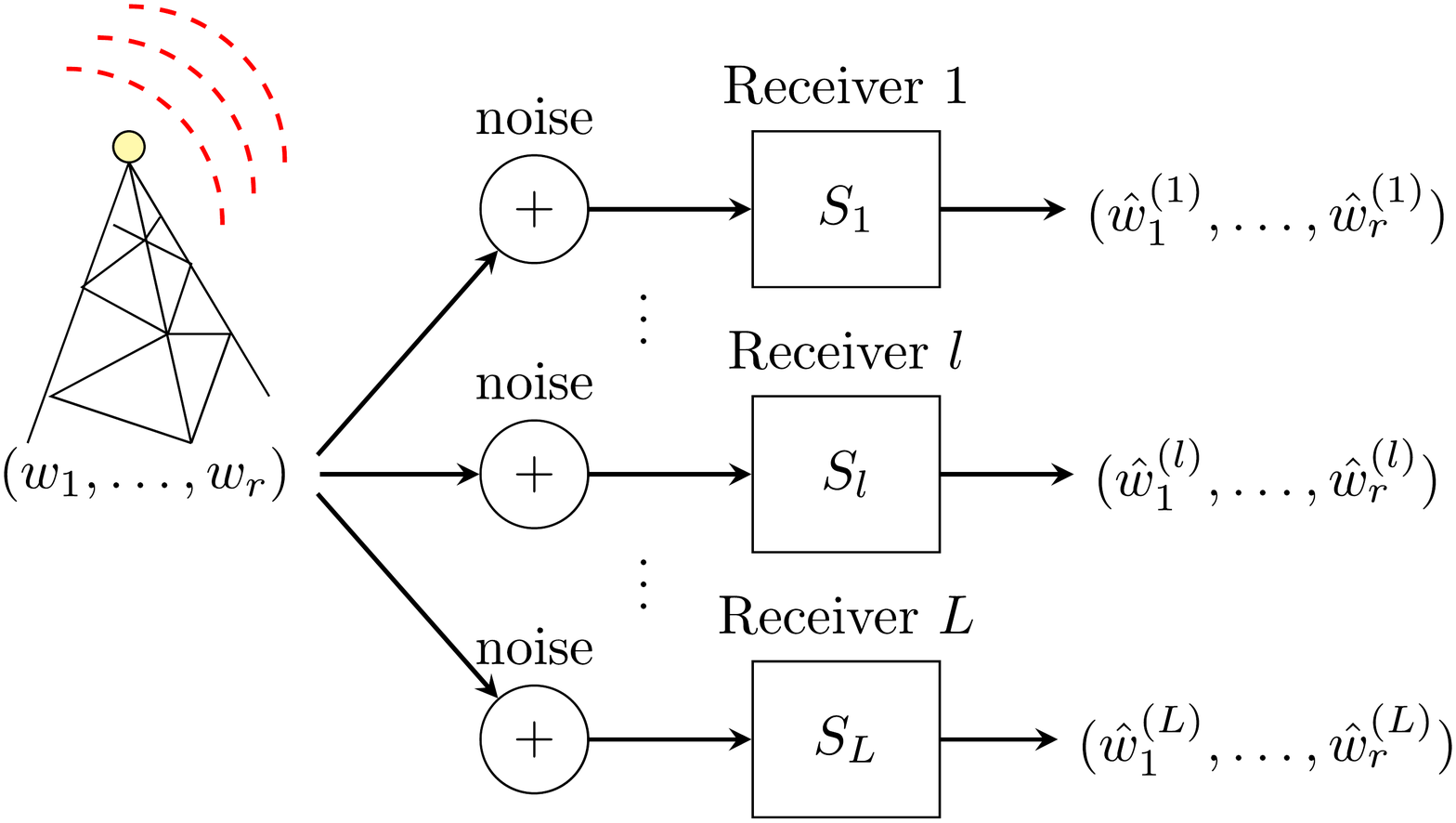}
\vspace{-.5cm}
\caption{Index coding over an AWGN channel with $r$ messages $\{w_1, \ldots, w_r\}$ and $L$ receivers. Each receiver $l$ estimates $(\hat{w}_1^{(l)}, \ldots, \hat{w}_r^{(l)})$ based on their received signal and prior knowledge $w_{S_l} \subset \{w_1, \ldots, w_r\}$, \cite{huang2017lattice}.}
\label{indexscheme}
\end{figure}

We summarise next the lattice index code construction using the Chinese Remainder Theorem of \cite{natarajan2015lattice2, natarajan2015lattice, Sue2018}. 

Given a lattice $\Lambda\subset\R^n$, let $p_1, \ldots, p_r \in \mathbb{Z}$ be $r$ distinct primes, and $q = \prod_{j=1}^r p_j$. The index code construction uses lattices $\Lambda_1, \ldots, \Lambda_r$, $\Lambda_j = \frac{q}{p_j}\Lambda$, $j = 1, \ldots, r$,  with the common sub-lattice $\Lambda' \subset \Lambda_j$, $\Lambda'=q\Lambda$ , and lattice constellations $\Lambda_1/\Lambda', \ldots, \Lambda_r/\Lambda'$ as the message alphabets, each with cardinality $p_j^n$. The encoding function is,
{\small \begin{align} 
\varphi: \Lambda_1/\Lambda' \times \cdots \times \Lambda_r/\Lambda' &\to \Lambda/\Lambda' \nonumber\\
(w_1, \ldots, w_r) &\mapsto x = (w_1 + \cdots + w_r) \mod \Lambda', \label{lic}
\end{align}}

\noindent where $w_j \in \Lambda_j/\Lambda'$ and $|\Lambda/\Lambda'| = q^n$. This mapping relates to the isomorphism provided by the Chinese Remainder Theorem, \cite{natarajan2015lattice} shows that if $\Lambda$ is the densest lattice in dimension, this construction achieves a uniform side information gain of $\approx 6.02$ dB/bit/dim.

In \cite{natarajan2015lattice} it is also considered constructions using primes over Gaussian, Eisenstein and Hurwitz quaternion integers. In \cite{huang2017lattice}, \textit{Huang} generalizes this lattice index coding scheme to rings of algebraic integers, enabling designs for scenarios with Rayleigh fading channels. 
An extension using cyclic division algebras is proposed in \cite{huang2018layered}.

\section{CRT lattice index coding}

In this section, it is introduced the lattice index coding scheme derived from Construction $\pi_A$ and presented its main properties.

\subsection{Main properties}

Let $p_1,\ldots,p_r\in\Z$ be distinct primes and let $q=\prod_{j=1}^{r}p_j$. Consider $\cod_j\subset\Z_{p_j}^n$ linear codes with $\text{rank}(\cod_j)=k_j$, $j=1,\dots r$. We will denote as $\Lambda_j=\Lambda_A(\cod_j)$ and $\Lambda = \Lambda_{\pi_A}(\cod)$, the Construction A and Construction $\pi_A$ lattice of each code $\cod_j$ and $\cod$, respectively, as in (\ref{constructionA}) and \textit{Definition \ref{construcao_pia}}. The bijective map $\phi^{-1}$, in \textit{Definition \ref{construcao_pia}}, restricted to the product $\cod_1\times \dots \cod_r$ can be written as,
\begin{align}
\Psi:& \ \Lambda_1/p_1\Z^n\times\ldots\times\Lambda_r/p_r\Z^n\rightarrow\Lambda/q\Z^n \nonumber\\
&(w_1, \ldots, w_r) \mapsto \ \sum_{j=1}^r x_jm_jw_j\mod q\Z^n.
\label{pia_restriction}
\end{align}

\noindent with $x_j$ obtained from Bèzout identity and $m_j=q/p_j$, $j=1,\ldots r$. We call a lattice code obtained by (\ref{pia_restriction}) as a \textit{CRT lattice index code}. Since the cardinality of each code is given by $|\cod_j|=p_j^{k_j}, \ j=1,\ldots,r$, then  the cardinality of $\cod=\Psi(\cod_1\times\cdots\times\cod_r)$ becomes $|\cod|=\prod_{j=1}^r|\cod_j| = \prod_{j=1}^r p_j^{k_j}$,  $\text{rank}(\cod)=\max\{k_j\}$ and the code $\cod\subset\Z_q^n$ is a free linear code if, and only if $\cod_j$ is free with the same rank $k$ for all $j$, \cite[Thm 2.4]{dougherty2017algebraic}.

For brevity in notation, let $S$ denote the set of side information at receiver $l$, we have that the rate of the $j^{th}$ message is given by,
{\small \begin{align*}
    R_j=\frac{1}{n}\log_2|\Lambda_j/p_j\Z^n|= \frac{1}{n}\log_2|\cod_j| = \log_2 p_j^{k_j/n} \text{bit}/\text{dim}.
\end{align*}}

The subcodes $\cod_{S}\subset\cod$ can be characterised as,
\begin{align*}
\cod_{S}=\{\Psi(w_1,\ldots,w_r); w_j=v_j, j\in S, w_j\in \cod, j\notin S\},    
\end{align*}

\noindent assuming that the receiver $l$ has knowledge of the side information $w_j=v_j$, for $j\in S$.

For example, when $S=\{1,\dots,s\}$ then $w_i=v_i$, for $i\in S$ we have \(\Psi (v_1,v_2,\ldots,v_s,w_{s+1},\ldots,w_r) = 
     \sum_{i=1}^s x_i m_i v_i \mod q\Z^n+\Psi (0,\ldots,0,w_{s+1},\ldots,w_r)
\)
\noindent what is a shifted version of $\Psi(0,\ldots,0,w_{s+1},\ldots,w_r)$. 

Denoting the complement of $S$ as $S^c$, let $\Lambda_{S^c}=\sum_{j\in S^c} \Lambda_j$ be the lattice obtained by the sum of $\Lambda_j, j\in S^c$, then, the subcode used for decoding at the receiver $l$ is,
\begin{equation*}
    \cod_{S}=\left\{\left(\sum_{j\in S} x_j m_j v_j + \Lambda_{S^c}\right) \mod q\mathbb{Z}^n \right\},
\end{equation*}

\noindent that is, $\cod_{S}$ is obtained by the translation of the Voronoi constellation $\Lambda_{S^c}/q\Z^n$ by the vector $\sum_{j\in S} x_j m_j v_j \mod q\mathbb{Z}^n$ and the minimum distance in $\cod_{S}$ is  $d_{S}=d_{\text{min}}\left(\Lambda_{S^c} \right)$.

The next proposition establishes the main properties of the CRT lattice index codes.

\begin{Prop}
\label{result}
For the lattices $\Lambda_1,\ldots, \Lambda_r$ and $\Lambda$ as in (\ref{pia_restriction}), the following is true, 
(i) $\text{vol}(\Lambda_{S^c})=\prod_{j\in S} p_j^{k_j}\text{vol}(\Lambda)$; (ii) $d_0\leq d_{S}\leq \prod_{j\in S} p_j d_0$, where $d_S = d_{\min}(\Lambda_{S^c})$ and $d_0 = d_{\min}(\Lambda)$.
\end{Prop}

\begin{proof} 
(i) Consider $S = \{s+1, \ldots, r\}$ and $S^c = \{1, \ldots, s\}$,
\[
\frac{\text{vol}(\Lambda_{S^c})}{\text{vol}(\Lambda)} = \frac{q^n / \prod_{j \in S^c} p_j^{k_j}}{q^n / \prod_{j=1}^r p_j^{k_j}} = \prod_{j \in S} p_j^{k_j}.
\]

(ii) Since $\sum_{j \in S^c} \Lambda_j \subset \Lambda$, it follows that $d_0 \leq d_S$. Moreover, the Construction A of $\prod_{j \in S} p_j \cdot \cod$ we have,
{\small \[
\Lambda_A\left(\prod_{j \in S} p_j \cdot \cod\right) = \prod_{j\in S} p_j \left( \Lambda_{S^c}\right) \subset \sum_{j \in S^c} \Lambda_j.
\]}
Hence,
\(
d_{\min}\left(\Lambda_A\left(\prod_{j \in S} p_j \cdot \cod\right)\right) = \prod_{j \in S} p_j \cdot d_0\)
and then,  \(d_S \leq \prod_{j \in S} p_j \cdot d_0.\)
\end{proof}

\begin{Obs}
    If we consider $k=n$, the CRT lattice index code introduced here is equivalent to those described in \cite{natarajan2015lattice}, specifically for $\Lambda = \Z^n$. Therefore, to explore cases not addressed in the literature, we will focus exclusively on the case $k < n$.
\end{Obs}

\subsection{An upper bound for side information gain}

To analyze the side information gain, let $S$ be the set of side information at receiver $l$ and $d_{S} = d_{\min}(\Lambda_{S^c})$ be the minimum distance of the subcode $\cod_S$. Since $R =\sum_{j=1}^r R_j = \frac{1}{n} \sum_{j=1}^r \log_2 |\Lambda_j / p_j \mathbb{Z}^n|$, it follows that,
\(
R = \frac{1}{n}\log_2 |\Lambda/q\mathbb{Z}^n|.
\)

The side information rate for this receiver $l$ is,
\(
R_S = R - \sum_{j \in S^c} R_j = \frac{1}{n} \log_2 \left(\frac{\text{vol}(\Lambda_{S^c})}{\text{vol}(\Lambda)}\right).
\)

As shown in \cite{natarajan2015lattice}, the centre densities of $\Lambda$ and $\Lambda_{S^c}$ allow us to derive a bound for the side information gain. If $\delta(\Lambda) \geq \delta(\Lambda_{S^c})$ for all $S$, then,
\(
\Gamma(\cod) = \min_S \frac{10 \log_{10}(d_{S}^2 / d_0^2)}{R_S} \leq 20 \log_{10} 2 \approx 6.02 \text{ dB/bit/dim}.
\)

However, $\delta(\Lambda) \geq \delta(\Lambda_{S^c})$ is not always true, but it is still possible to derive an upper bound. 

Consider the special case of CRT lattice index code when $\text{rank}(\cod_j)=k$ for all $j=1,\dots r$. Since, by \textit{Proposition \ref{result} (i)}, $\text{vol}(\Lambda_{S^c})/\text{vol}(\Lambda) = \prod_{j \in S} p_j^{k}$, then the side information rate is,
\(
R_S = \log_2 \left(\prod_{j \in S} p_j^{k / n}\right).
\) 
From \textit{Proposition \ref{result} (ii)}, we obtain, $\log_{10}(d_S / d_0) \leq \log_{10} \left(\prod_{j \in S} p_j^{k / n}\right)$, then,
\begin{equation}
\Gamma(\cod) \leq 20 \log_{10} 2 \cdot \frac{\log_{10} \left(\prod_{j \in S} p_j\right)}{\log_{10} \left(\prod_{j \in S} p_j^{k / n}\right)} \leq \frac{n}{k} \cdot 20 \log_{10} 2.    
\label{crtbound}
\end{equation}

This provides an upper bound for the side information gain of a CRT lattice index coding derived from the same rank codes. A previous analysis of this bound, considering codes with different ranks, can be found in \cite{juliana2024indexcnmac}.

\section{CRT lattice index coding with uniform side information gain}

In this section, we construct a CRT lattice index code with uniform side information gain. To achieve this, we consider that all codes $\cod_j$ in this construction must have the same rank $k$. Additionally, we restrict the generators of the codes to ensure compatible minimum distances for any index set $S$.

First, we focus on codes generated by $k$ linearly independent vectors over $\Z_{p_j}^n$ with minimum distance $1$ and construct a CRT lattice index code $\cod$ such that $d_{\min}(\cod)=1$. 

One way to achieve this is as follows. Considering linear codes $\cod_1, \dots, \cod_r$ over $\mathbb{Z}_{p_1}^n, \dots, \mathbb{Z}_{p_r}^n$, each generated by $k$ canonical vectors with $\mathbf{e}_l$ a common generator for all codes. For each $\cod_j$, we have $d_{\min}(\cod_j) = 1, j=1,\dots r$. 

Let $\cod = \Psi(\cod_1\times\cdots\times\cod_r)$ be the CRT lattice index coding obtained by $\cod_1, \dots, \cod_r$. Since $\mathbf{e}_l$ belongs to all $\cod_j$, the minimum distance of $\cod$ is achieved with $c = \sum_{i=1}^r x_i m_i \mathbf{e}_l \mod q = \mathbf{e}_l$, resulting in $d_{\min}(\cod) = 1$. This code achieves the upper bound for side information gain in (\ref{crtbound}). In fact, for any subset $S$, $d_{\min}(\cod_S)=1$ then we have,
\[
d_S = \prod_{j \in S} p_j \cdot d_{\min}(\cod_S) = \prod_{j \in S} p_j,
\]
and the corresponding side information rate is,
\(
R_S = \log_2\left(\prod_{j \in S} p_j^{k/n}\right) = \frac{k}{n} \log_2\left(d_S/d_0\right).
\)
Thus, for all $S$, the side information gain is,
\begin{align*}
\Gamma(\cod) &= \frac{20 \log_{10}(d_S/d_0)}{(k/n)\log_2(d_S/d_0)} = \frac{n}{k} \cdot 20 \log_{10} 2.    
\end{align*}

Note that for sufficiently large dimension and using codes with rank $k=n-1$, when $n\rightarrow\infty$ we approach $\approx 6.02$ dB/bit/dim as expected.

A different analysis for uniform side information gain can be made for rank-$1$ codes over $\mathbb{Z}_q^n$. In this case, we will focus on numbers expressible both as a sum of $N$ squares and as a product of primes, where each prime $p_j$ can also be written as a sum of $N$ squares, i.e.,
\begin{equation}
q = \sum_{i=1}^N x_i^2 = \prod_{j=1}^r p_j, \quad p_j = \sum_{i=1}^N (a_i^{(j)})^2.    
\label{sumquares}
\end{equation}

This problem relates to the multiplicative quadratic form problem \cite[Chap. 2]{Pfister_1995}, and has explicit solutions for $N=2,4,8$. We consider the CRT lattice index code from codes $\cod_j = \langle g_j \rangle = \langle (x_1, \dots, x_N) \mod p_j \rangle \subset \mathbb{Z}_{p_j}^N$. The sum of squares condition with other requirements ensures uniform side information gain in CRT lattice index coding, which depends on the number of squares considered as shown next.

\begin{Lema} \label{uniform_construction}
    Let $q$ be an integer that can be expressed as a sum of $N$ squares, i.e.,  $q = x_1^2 + \dots + x_N^2$. Then for the linear code $\mathcal{C} = \langle (x_1, \dots, x_N) \rangle$ we have $d_{\min}(\mathcal{C}) = \sqrt{q}$.
\end{Lema}

\begin{proof}
Let $g = (x_1, \dots, x_N)$, with $\|g\|^2= x_1^2 + \dots + x_N^2 = q$.  We have that $\cod=\{x = u\cdot g; u\in\Z_q\}\subset\Z_q^N$ satisfies $d_{\min}(\cod) = \sqrt{q}$ because $\|x\|^2 \mod q= 0 $ then $\min\|x\|^2 = q$ achieved for $x = g$.
\end{proof}

An additional requirement is that for the codes $\cod_j$, $d_{\min}(\cod_j)=\sqrt{p_j}$. To assure this we include the condition that $\cod_j=\langle (x_1, \dots, x_N) \mod p_j \rangle = \langle (a_1^{(j)},\dots a_n^{(j)})$, what is equivalent to have solution for the following systems of congruences with $\lambda_j \in \mathbb{Z}_{p_j}$,
{\small \begin{equation}
    \label{system}
    (a_1^{(j)},\dots, a_N^{(j)})\equiv \lambda_j (x_1,\dots, x_N)\mod p_j \quad \forall \quad j=1,...,r.
\end{equation}}

\begin{Ex}
    Consider $N=4$, $q=133=7 \cdot 19$ and note that $133 = 1^2 + 2^2 + 8^2 + 8^2 = 5^2 + 6^2 + 6^2 + 6^2$. 
    
    Assuming $133 = 1^2 + 2^2 + 8^2 + 8^2$, we have no solution for $\lambda_2$ by considering any choice of decompositions $(a_1^{(2)})^2 + (a_2^{(2)})^2 + (a_3^{(2)})^2 + (a_4^{(2)})^2 = 19$. For $133 = 5^2 + 6^2 + 6^2 + 6^2$, the system (\ref{system}) has solution $\lambda_1 = 1$ or $-1$ what gives $(a_1^{(1)}, a_2^{(1)}, a_3^{(1)}, a_4^{(1)}) = (-2,-1,- 1,- 1)$ or $(2, 1, 1, 1)$, and $\lambda_2 = 3$ or $-3$ giving $(a_1^{(2)}, a_2^{(2)}, a_3^{(2)}, a_4^{(2)}) = (-4, -1, -1, -1)$ or $(4, 1, 1, 1)$. Thus, we can take $\cod=\langle (5,6,6,6)\rangle\subset\Z_{133}^4$, $\cod_1 = \langle (2,1,1,1)\rangle\subset\Z_7^4$ and $\cod_2=\langle (4,1,1,1)\rangle\subset\Z_{19}^4$.
\end{Ex}

\begin{Obs}
\label{dimension_greater}
We can also consider a Cartesian product of the previous rank$-1$ codes (\ref{sumquares}) to construct CRT lattice index codes with uniform side information gain. By taking the Cartesian product of $m$ codes $\psi: \cod_j \times \cdots \times \cod_j \rightarrow \tilde{\cod}_j \subset \mathbb{Z}_{p_j}^{mN}$, we obtain $\tilde{\cod}_j = \langle g_{j1}, \dots, g_{jm} \rangle$ where $g_{ji} = (0,\dots, 0, x_1, \dots, x_N, 0,\dots, 0)\mod p_j$. Therefore, we can use the codes $\tilde{\cod}_1 \subset \mathbb{Z}_{p_1}^{mN}, \dots, \tilde{\cod}_r \subset \mathbb{Z}_{p_r}^{mN}$ and the map $\Psi$ to obtain $\tilde{\cod} = \Psi(\tilde{\cod}_1 \times \cdots \times \tilde{\cod}_r)\subset\Z_q^{mN}$. The code $\tilde{\cod}$ in dimension $mN$ has $\text{rank}(\tilde{\cod})=m$. For simplicity, we will denote the codes $\tilde{\cod}$ and $\tilde{\cod}_j$ obtained by the Cartesian product as $\cod$ and $\cod_j$, respectively, with the dimension understood from the context.
\end{Obs}

To achieve uniform side information gain in the above construction we will need that for any side information set $S$ the subcodes $\cod_S$ have minimum distance $\sqrt{\prod_{j\in S^c} p_j}$. We then require that the sum of squares (\ref{sumquares}) and the systems (\ref{system})  be valid for all possible products of $p_j, j=1,\dots, r$.

\begin{Teo}\label{sidegainrank1}
Let $p_j\in \mathbb{Z}$ be primes, $j\in I=\{1,...,r\}$, and let $q = \prod_{j=1}^r p_j = x_1^2 + \dots + x_N^2$ 
assume that the systems 
    {\small \[
     \left\{\begin{array}{l}
    (b_1^{(k)})^2 + \dots + (b_N^{(k)})^2 = \displaystyle\prod_{\{k_1,\dots, k_L\}\subset I} \hspace{-.6cm} p_{k_1}\dots p_{k_L}\\
        \lambda_k (b_1^{(k)}, \dots, b_N^{(k)}) \equiv (x_1, \dots, x_N) \mod\hspace{-.6cm} \displaystyle\prod_{\{k_1,\dots, k_L\}\subset I} \hspace{-.6cm} p_{k_1}\dots p_{k_L},    \end{array} \right.\] }
    
\noindent have a solution for $\lambda_k \in \mathbb{Z}$ with $k=1,\dots, 2^r-1$. Consider $n = mN$, the codes $\cod_j \subset \mathbb{Z}_{p_j}^n$, with $\text{rank}(\cod_j) = m$, and the CRT lattice index code $\cod = \Psi(\cod_1 \times \cdots \times \cod_r) \subset \mathbb{Z}_q^n$, with $\text{rank}(\cod) = m$. Then, the uniform side information gain is,
\[
\Gamma(\mathcal{C}) = \frac{N}{2} \cdot 20 \log_{10} 2 \ \text{dB/bit/dim}.
\]
\end{Teo}

\begin{proof}
Without loss of generality, let $S = \{s+1, \dots, r\}$. By \textit{Lemma \ref{uniform_construction}}, the minimum distance of $\cod_{S}$ is $d_{\min}(\cod_{S}) = \prod_{j \in S^c} p_j^{1/2}$. 
The distance $d_S$ is, 
\[
d_S = d_{\min}(\Lambda_{S^c}) = \prod_{j \in S} p_j \cdot \prod_{j \in S^c} p_j^{1/2}.
\]

Therefore we can derive the side information rate,
{\small \[
R_S = \log_2\prod_{j \in S} p_j^{m/mN} = \frac{2}{N}\log_2\prod_{j \in S} p_j^{1/2} = \frac{2}{N} \log_2\left(\frac{d_S}{d_0}\right),
\]}
obtaining the following uniform side information gain,
\[
\Gamma(\cod) = \frac{20 \log_{10}(d_S/d_0)}{(2/N) \log_2(d_S/d_0)} = \frac{N}{2} \cdot 20 \log_{10} 2.
\]
\end{proof}

Note that for the sum of two squares ($N=2$), we achieve the side information gain of $\approx 6.02$ dB/bit/dim (see \textit{Remark} \ref{finalremark}). For $N>2$, we get a greater side information gain in a higher dimensional lattice code. It should be also expected a worse sphere packing density in $\Lambda$ when compared to $\Lambda_{S^c}$, as remarked for the index codes in \cite{natarajan2015lattice}.

\begin{Ex}
\label{exemplo}
Consider $m=1, N=3$, the prime numbers $p_1=3$, $p_2 = 11$ and $p_3=17$ and their product $q=561=13^2 + 14^2 + 14^2$. Let $\mathcal{C}_1 = \langle (13, 14, 14)\mod 3) \rangle = \langle (1, -1, -1) \rangle$, $\mathcal{C}_2 = \langle \ 4\cdot (13, 14, 14) \mod 11 \rangle = \langle (-3, 1, 1) \rangle$ and $\mathcal{C}_3 = \langle \ 5\cdot (13, 14, 14) \mod 17 \rangle = \langle (-3, 2, 2) \rangle$. The codes $\cod_1, \cod_2$ and $\cod_3$ have minimum distance $d_{\min}(\cod_1) = \sqrt{3}, d_{\min}(\cod_2) = \sqrt{11}$ and $d_{\min}(\cod_3)=\sqrt{17}$, respectively. The code $\cod$ has minimum distance $d_{\min}(\cod) = \sqrt{561}$. Table \ref{distance_ex6} describes the minimum distance, message rate, and information gain for all non-trivial index sets. Figure \ref{ser_ex6} shows the SNR versus symbol error rate. Note that for any subset $S$, we always achieve the same side information gain.

\begin{table}[htpb]
        \centering
        \begin{tabular}{c|c|c|c}
          Index set   &  Minimum distance & Message rate & $\Gamma(\cod, S)$\\
          \hline
          \hline
     $S = \{1\}$  & $3\sqrt{187}$ & $1/3 \log_2 3$ & 9.03 \\
     $S = \{2\}$ & $11\sqrt{51}$ & $1/3 \log_2 11$ & 9.03 \\
     $S = \{3\}$ & $17\sqrt{33}$ & $1/3 \log_2 17$ & 9.03 \\
     $S = \{1,2\}$ & $33\sqrt{17}$ & $1/3 \log_2 33$ & 9.03 \\
     $S = \{1,3\}$ & $51\sqrt{11}$ & $1/3 \log_2 51$ & 9.03 \\
     $S = \{2,3\}$ & $187\sqrt{3}$ & $1/3 \log_2 187$ & 9.03 
     \vspace{.1cm}
        \end{tabular}
        \caption{Minimum distance, message rate, and side information gain in each subcode $\cod_S$.}
        \label{distance_ex6}
    \end{table}
    \begin{figure}[htpb]
		\centering
    		\includegraphics[scale=.17]{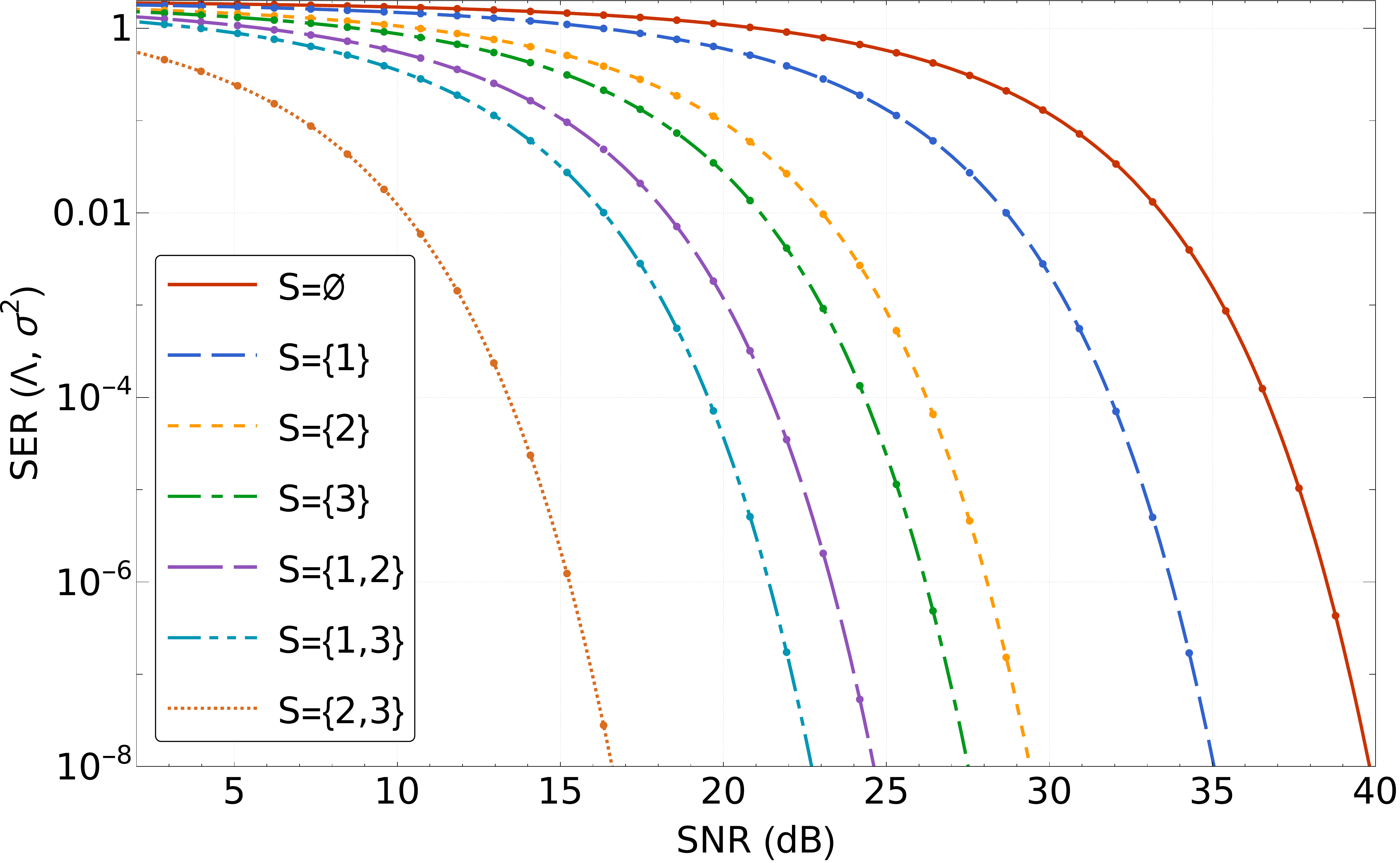} 
		\caption{SNR versus symbol error rate over the AWGN channel for the code of \textit{Example \ref{exemplo}}.}
		\label{ser_ex6}
\end{figure}
\end{Ex}

\begin{Obs} \label{finalremark}
Consider the lattice obtained via Construction $\pi_A$ from the codes \(\mathcal{C}_1 = \langle (a_1, b_1) \rangle \subset \mathbb{Z}_{p_1}^2\) and \(\mathcal{C}_2 = \langle (a_2, b_2) \rangle \subset \mathbb{Z}_{p_2}^2\), where \(a_j^2 + b_j^2 = p_j\), $j=1,2$. This setup allows us to obtain the one dimensional code $\cod = \Psi(\cod_1\times\cod_2) = \langle a,b\rangle\subset\Z_q^2$, with $q = a^2 + b^2 = (a_1a_2 -b_1b_2)^2 + (a_2b_1 + a_1 b_2)^2$. The code $\cod$ is given by the quotient of the lattices $\Lambda'/q\Z^2$, where $\Lambda'=\Lambda_{\pi_A}(\cod)$ is generated by the set $\{(a,b),(a^2+b^2,0),(0,a^2+b^2)\}$,\cite{Sue2018}, which has basis $\{(a,b), (-b,a)\}$. It should be noticed that this CRT lattice index code is a rotated scaled version of the lattice index code using Gaussian integers $\Z[\ii]$, presented in \cite{natarajan2015lattice}. Identifying $\Z[\ii]$ with the lattice $\Z^2$, in (\ref{lic}) we have \(\varphi(\Lambda_1/\Lambda'\times\Lambda_2/\Lambda') = \Z^2/\Lambda',\) with $\Lambda_j = \langle(a_j, b_j), (-b_j, a_j)\rangle, j=1,2$.

\vspace{-.1cm}
\begin{figure}[htpb]
    \centering
    \includegraphics[width=0.49\linewidth]{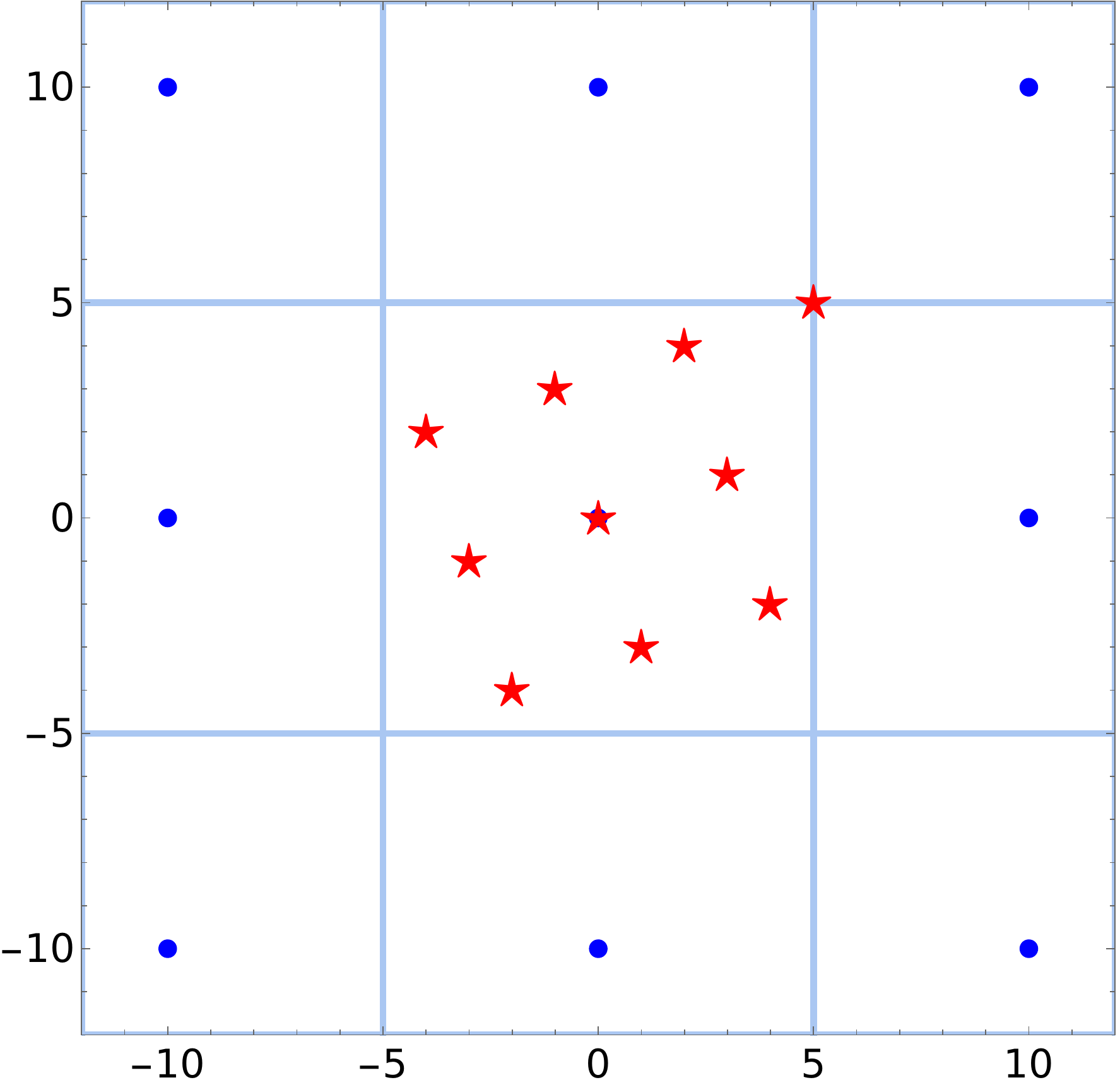}
    \includegraphics[width=0.49\linewidth]{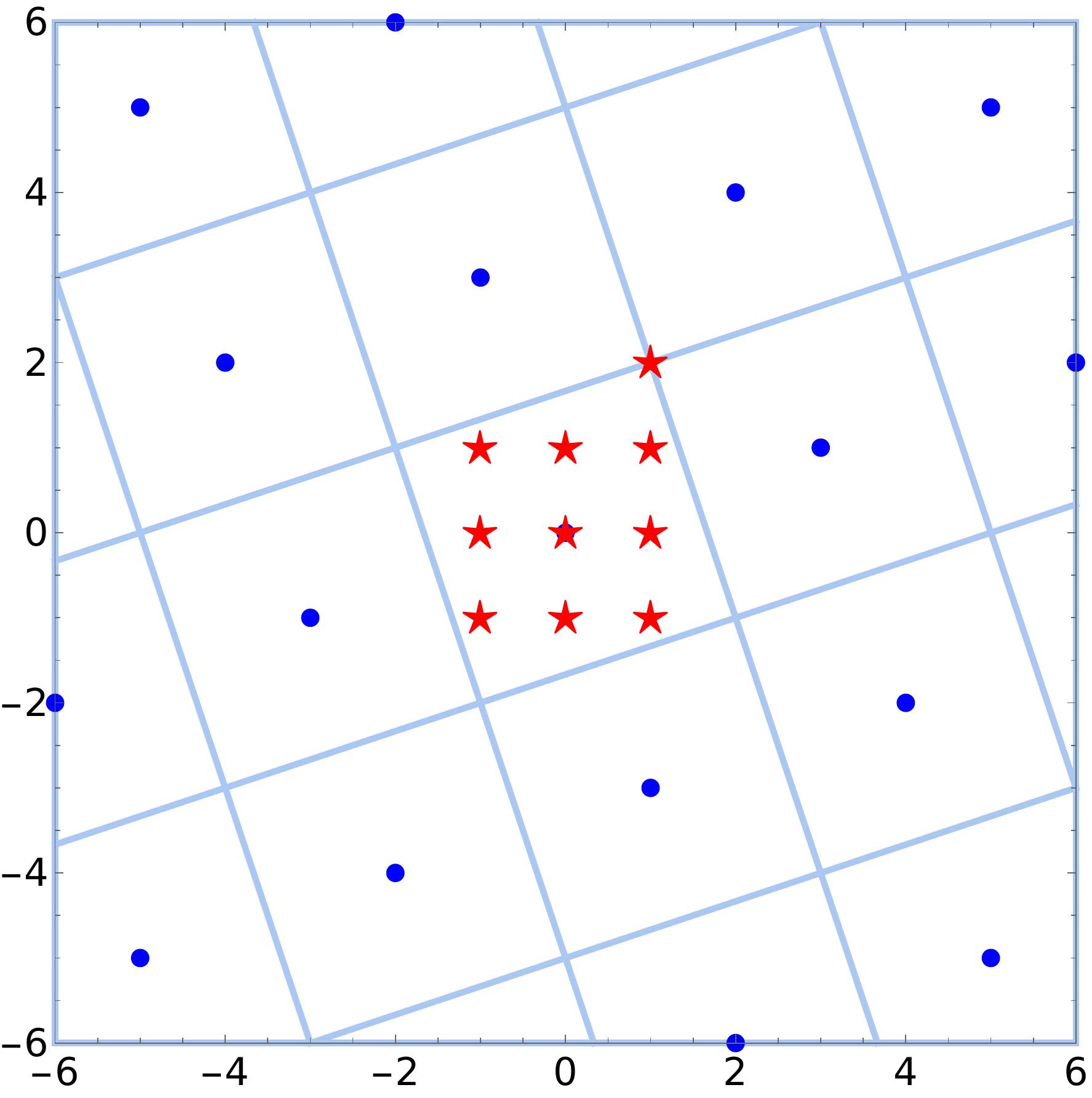}
    \caption{On the left, an example of the CRT lattice index code \(\Psi(\cod_1 \times \cod_2) = \Lambda' / q\mathbb{Z}^2\) where $\cod_1 = \langle(1,1)\rangle\subset\Z_2^2$ and $\cod_2 = \langle (1,2)\rangle\subset\Z_5^2$. On the right, the associated lattice index code from Gaussian integers, \(\varphi(\Lambda_1 / \Lambda' \times \Lambda_2 / \Lambda') = \mathbb{Z}^2 / \Lambda'\), as constructed in \cite{natarajan2015lattice}.}
    \label{examplecodes}
\end{figure}
\end{Obs}

\vspace*{-.5cm}
\section{Conclusion and Perspectives}

In this work, we explore the application of Construction $\pi_A$ lattices over $\Z$ to the index coding problem, presenting the CRT lattice index coding scheme. This approach can take advantage of the multilevel nature of Construction $\pi_A$ and enlarge the scenario of index coding over $\Z$ as considered in \cite{natarajan2015lattice}. It is shown an upper bound for the side information gain when the Construction $\pi_A$ is considered over same rank codes and identified some codes that achieve this gain. It is established conditions that assure uniform side information gain. 

Perspectives for future work, include to investigate index codes using Construction $\pi_A$ over other rings of integers such as Gaussian integers, Eisenstein integers, and maximal orders in quaternion algebras to determine conditions required for achieving uniform side information gain and bounds that can be obtained using these structures.

\section*{Acknowledgment}
This work was partially supported by Coordination for the Improvement of Higher Education Personnel (CAPES - Financial Code 001), by São Paulo Research Foundation (FAPESP), Brazil, 2020/09838 -0 and by Brazilian National Council for Scientific and Technological Development (CNPq), 308399/2023-4.


\bibliographystyle{IEEEtran}
\bibliography{bibliography}

\end{document}